\documentclass[final,nomarks]{macros/dmtcs-episciences}

\author{Kehinde Adeogun \and Christos Kapoutsis}
\affiliation{Carnegie Mellon University in Qatar}

\title
[A quadratic lower bound for 2DFAs against one-way liveness]
{A quadratic lower bound for\\2DFAs against one-way liveness%
\thanks{Supported by Qatar Foundation, via CMUQ's Seed Research program and project \textit{Minicomplexity}. All statements made herein are the responsibility of the authors.}}


\publicationdata
{vol. 28:4 special issue for SOFSEM 2026}
{2026}
{?}
{10.46298/dmtcs.17609}
{2000-1-1; None}
{2000-1-1}

\usepackage{pdflscape}
\usepackage{bbold}
\usepackage{tgtermes}

\usepackage{amstext}               
\usepackage[greek,english]{babel}  
\usepackage[utf8x]{inputenc}       

\newcommand{\MACH}[1]{\text{\scshape\mdseries\rmfamily #1}}
\newcommand{\MACHNUM}[1]{\oldstylenums{\upshape#1}}
\newcommand{\MACHs}{\text{\upshape\mdseries\rmfamily s}}

\newcommand{\fa  }{\MACH{fa}}  

\newcommand{\MACHdeterministic         }{\MACH{d}}

\newcommand{\MACHnondeterministic      }{\MACH{n}}

\newcommand{\MACHoneway   }{\MACHNUM{1}}

\newcommand{\MACHtwoway   }{\MACHNUM{2}}

\newcommand{\onfa   }{{\MACHoneway\MACHnondeterministic\fa}}

\newcommand{\onfas  }{{\onfa\MACHs}}

\newcommand{\tdfa   }{{\MACHtwoway\MACHdeterministic\fa}}

\newcommand{\tnfa   }{{\MACHtwoway\MACHnondeterministic\fa}}

\newcommand{\tdfas  }{{\tdfa\MACHs}}

\newcommand{\tnfas  }{{\tnfa\MACHs}}

\usepackage{amssymb}        
\usepackage{amsmath}        

\newcommand{\bbB}{\mathbb{B}}

\newcommand{\ggS}{\varSigma}

\newcommand{\ggV}{\varOmega}

\newcommand{\gga}{\alpha}
\newcommand{\ggb}{\beta}

\newcommand{\ggd}{\delta}
\newcommand{\gge}{\epsilon}

\newcommand{\ggu}{\vartheta}

\newcommand{\ggl}{\lambda}

\newcommand{\something}{{\,\cdot\,}}
\newcommand{\kkk}{{,}}

\newcommand{\ppp}{{+}}
\newcommand{\mmm}{{-}}
\newcommand{\xxx}{{\times}}

\newcommand{\atmost}{{\leq}\,}
\newcommand{\atleast}{{\geq}\,}

\newcommand{\subs }{\subseteq }
\newcommand{\sups }{\supseteq }

\newcommand{\s}{^*}
\newcommand{\z}{_*}

\newcommand{\fml}[2][1]{(#2_h)_{h\geq#1}}

\newcommand{\start  }{_\textup{s}}

\newcommand{\accept }{_\textup{a}}

\newcommand{\reject }{_\textup{r}}

\newcommand{\lend}{{\vdash}}
\newcommand{\rend}{{\dashv}}
\newcommand{\lrends}{\ensuremath{\{\lend\kkk\rend\}}}
\newcommand{\lrend}[1]{\mbox{$\lend#1\rend$}}

\newcommand{\tagL   }{\textup{\textsc{l}}}

\newcommand{\tagR   }{\textup{\textsc{r}}}
\newcommand{\tagLR  }{\textup{\textsc{lr}}}
\newcommand{\tagRL  }{\textup{\textsc{rl}}}
\newcommand{\tagsLR }{\ensuremath{\{\tagL\kkk\tagR\}}}

\newcommand{\subLR}{_\textup{\textsc{lr}}}
\newcommand{\subRL}{_\textup{\textsc{rl}}}

\newcommand{\comp  }{\textsc{comp}}
\newcommand{\lcomp }{\textsc{lcomp}}
\newcommand{\rcomp }{\textsc{rcomp}}

\usepackage{graphicx}             
\usepackage[usenames]{color}      

\newcommand{\figdir}{figures}

\graphicspath{{\figdir/}}

\usepackage{amstext}        
\usepackage{xstring}        

\newcommand{\PROB    }[1]{\text{\scshape\mdseries\rmfamily #1}}
\newcommand{\PROBMATH}[1]{\text{\scriptsize$#1$}}

\newcommand{\tttowl         }{\PROBMATH{3}\PROB{owl}}
\newcommand{\nowl     }[1][n]{\IfInteger{#1}{\PROBMATH{#1}}{\ensuremath{#1}}\PROB{owl}}

\newcommand{\owl            }{\PROB{owl}}

\usepackage{amstext}               
\usepackage[greek,english]{babel}  
\usepackage[utf8x]{inputenc}       

\newcommand{\CLSS   }[1]{\text{\upshape\mdseries\sffamily #1}}
\newcommand{\CLSSNUM}[1]{\CLSS{#1}}

\newcommand{\CLSSnondeterministic      }{\CLSS{N}}

\newcommand{\CLSSoneway   }{\CLSSNUM{1}}

\newcommand{\onen   }{{\CLSSoneway\CLSSnondeterministic}}

\usepackage{amsthm}    

\theoremstyle{plain}
\newtheorem{lemma}{Lemma}
\newtheorem{claim}{Claim}
\newtheorem{theorem}{Theorem}

\theoremstyle{definition}
\newtheorem{definition}{Definition}

\theoremstyle{remark}
\newtheorem{example}{Example}

\newcommand{\qqed}{$\boxdot$}

\newcommand{\ID}{\text{id}}
\newcommand{\bbBh}{\ensuremath{\bbB^{h\times h}}}
\newcommand{\ONE}{\mathfrak{1}}
\newcommand{\ZERO}{\mathfrak{0}}
\newcommand{\ZEROh}{\mathrm{0}_h}
\newcommand{\IDTYh}{\mathrm{I}_h}

\begin{document}
\maketitle

\begin{abstract}
We show that every \textit{two-way deterministic finite automaton} ({\oldstylenums{\upshape 2}}{\scshape\mdseries\rmfamily dfa}) that solves \textit{one-way liveness} on height~$h$ has $\ggV(h^2)$~states. This implies a quadratic lower bound for converting \textit{one-way nondeterministic finite automata} to {\oldstylenums{\upshape 2}}{\scshape\mdseries\rmfamily dfa}s, which asymptotically matches Chrobak's well-known lower bound for this conversion on unary languages. In contrast to Chrobak's simple proof, which relies on a {\oldstylenums{\upshape 2}}{\scshape\mdseries\rmfamily dfa}'s inability to differentiate between any two sufficiently distant locations in a unary input, our argument can be applied to inputs over any alphabet and is structured around a main lemma that is general enough to potentially be reused elsewhere.
\end{abstract}

\section{Introduction}
\label{sec:introduction}

Back in the 1970s, Sakoda and Sipser \cite{sasi78} (but also Seiferas~\cite{se73}) asked a fundamental question, which to this day remains an important open problem in automata and complexity theory: Can every \textit{two-way nondeterministic finite automaton} (\tnfa) with some number of states~$s$ be simulated by a \textit{two-way deterministic finite automaton} (\tdfa) with a number of states~$f(s)$ which is polynomial in~$s$? 

A standard way to approach this question is to study variant cases: cases where either the \tdfa\ is enhanced or the \tnfa\ is restricted, if we want to prove an upper bound for the respective variant~$f(s)$; or cases where either the \tdfa\ is restricted or the \tnfa\ is enhanced, if we want to prove a lower bound for the respective variant~$f(s)$.

Indeed, a \textit{polynomial} upper bound is known when the simulating \tdfa\ is enhanced into some restricted single-tape Turing machine: a \textit{weight-reducing} one (i.e., one that overwrites symbols only by ``lighter'' ones) or a \textit{bounded} \&\ \textit{linear-time} one (i.e., a linear-time one that writes only within the input)~\cite{gppp22}. Likewise, a \textit{quasi-polynomial} upper bound is known when the simulated \tnfa\ is restricted to \textit{unary} (i.e., accepting only unary inputs)~\cite{gemepi03}; or \textit{letter-bounded} (i.e., accepting only inputs of the form $a_1\s a_2\s\cdots a_k\s$ for fixed~$k$)~\cite{cpp25}; or \textit{outer-nondeterministic} (i.e., using nondeterminism only on the endmarkers)~\cite{gegupi14}. 

On the flip side, an \textit{exponential} lower bound is known when the simulating \tdfa\ is restricted to \textit{single-pass} (i.e., halting upon reaching an endmarker)~\cite{se73}; or \textit{sweeping} (i.e., reversing only on the endmarkers)~\cite{si80b}; or \textit{(almost) oblivious} (i.e., following sub-linearly many trajectories per input length)~\cite{hrsc03}; or \textit{few-reversal} (i.e., reversing sub-linearly often)~\cite{ka13iacc}; or \textit{logical} (i.e., satisfying some formula on ``reachability variables'' in each state)~\cite{hkks13,bihrko15,ra19}; and an \textit{infinite} lower bound is known when the \tdfa\ is a \textit{mole} (i.e., it only explores the configuration graph of the \tnfa)~\cite{ka07jalc}. Likewise, an \textit{exponential} lower bound is known when the simulated \tnfa\ is enhanced into a \textit{single-pebble} \tnfa\ (i.e., one equipped with a pebble that it can drop on and lift from an input cell)~\cite{geis09dcfs}; or into a \textit{1-limited automaton} (i.e., a single-tape nondeterministic Turing machine that  can overwrite an input cell when it first visits it)~\cite{pi19ifip}, even a \textit{once-marking} one (i.e., one that can only place a mark once on a single cell of the input)~\cite{pipr23afl}.


In this landscape, a natural question is: \textit{What about lower bounds for unrestricted \tdfas?} 

There are only two results of this kind and, quite counterintuitively, both are obtained in scenarios where the simulated \tnfa\ is (not in any way enhanced, but instead) severely restricted! 

The first one is Chrobak's well-known bound $f(s)=\ggV(s^2)$ for when the simulated \tnfa\ is both \textit{unary} (as above) and \textit{one-way} (i.e., moving its input head only forward; namely, a standard \textit{one-way nondeterministic finite automaton} or \onfa) \cite[Thm.~6.3]{ch86} (cf.\ also~\cite{to09}). Expectedly, the simple proof of this result relies heavily on the inability of the simulating \tdfa\ to differentiate between distant locations on the unary input. Moreover, a matching upper bound \cite[Thm.~6.2]{ch86} shows that no better lower bound is possible in this special case.     

The second result is the bound $f(s)=\ggV(s^2/\log s)$ of~\cite[Lemma~7]{ka18} for the case where the simulated \tnfa\ is both \textit{one-way} (again) and \textit{accepting only inputs of length~$3$}.\footnote{This is (easily) an equivalent restatement of \cite{ka18}'s lower bound for~\tttowl.} The (counting) argument for this second result is also quite simple. Moreover, here too, the possibility of a better lower bound is removed by a (quite non-trivial) matching upper bound~\cite[Lemma~8]{ka18}.   

Our present study offers a third lower bound for unrestricted \tdfas. 
As in the two existing results, we focus on \tnfas\ that are \textit{one-way}, i.e., \onfas; and, like Chrobak, we prove $f(s)=\ggV(s^2)$. Unlike the two existing results, however, our \onfas\ work on alphabets of any size and inputs of any length, which forces us into a much deeper analysis of the relevant computations. 

We express our argument in terms of \textit{one-way liveness} (\owl), the complete language (family) for the conversion of \onfas\ to \tdfas\ \cite[Thm.~2.3]{sasi78}. Specifically, we prove that \textit{every \tdfa\ that solves \owl\ on height~$h$ has $\ggV(h^2)$ states}, which directly implies that $f(s)=\ggV(s^2)$, since $\owl$ on height~$h$ is solved by a \onfa\ with $h$~states \cite[Thm.~2.3i]{sasi78}.    

After some preliminaries (Sect.~\ref{sec:preparation}), we recall (from~\cite{ka07jalc}) the basic definitions and facts about \textit{generic strings} (Sect.~\ref{sec:generic-strings}); and show how these force a \tdfa\ to use its states on their outer boundaries (Sect.~\ref{sec:exit-sizes}). We then prove our main lemma (Lemma~\ref{lem:main}), that a suitably selected pair of properties can force a \tdfa\ to spend $\atleast1$ state just to distinguish between them. Finally, we identify a quadratically long sequence of properties where every two successive ones are as described in the lemma (Sect.~\ref{sec:the-lower-bound}). This then implies the promised lower bound.

An earlier version of this work appeared in~\cite{adka26sofsem}. Here, we offer a strictly stronger quadratic lower bound ($\frac12\binom{h+1}{2}$, instead of only $\frac12\binom{h}{2}$), as we have the space to also include Stage~2 and Stage~3 of our full construction (in Section~\ref{sec:connectivities}) and the analysis of the respective connectivities (in Section~\ref{sec:basic-facts}). We also offer a better exposition, with more discussion, intuition, and examples.

\section{Preparation}
\label{sec:preparation}

For $n\geq1$, let $[n]:=\{1,2,\dots,n\}$. The set~$\ggS\s$ consists of all finite strings over~$\ggS$. If $z\in\ggS\s$, we let $|z|$, $z_i$, $z^i$ be its length, its $i$th~symbol ($i\in [|z|]$), and the concatenation of $i$ copies of itself ($i\geq0$). If $f:S\to S$, then $f^i$ is the composition of $i$~copies of~$f$ ($i\geq0$). We let $\bbB:=\{0\kkk1\}$ be the set of Boolean values and $\bbB^{n\times m}$ be the set of all $n\times m$ Boolean matrices ($n,m\geq1$).

\subsection{Two-way finite automata}
\label{sec:2dfas}

A \textit{two-way deterministic finite automaton} (\tdfa) is a tuple of the form $M=(Q,\ggS,\ggd,q\start,q\accept,q\reject)$, where $Q$~is a set of \textit{states}; $\ggS$~is an \textit{input alphabet}; the states $q\start\kkk q\accept\kkk q\reject\in Q$ are the \textit{start}, \textit{accept}, and \textit{reject} states, respectively; and $\ggd:Q\times(\ggS\cup\lrends)\to Q\times\tagsLR$ is a total \textit{transition function}, where $\lend\kkk\rend\not\in\ggS$ are the left and right endmarkers, and $\tagL\kkk\tagR$ are the two directions. An input $z \in \ggS\s$ is presented to $M$ between the endmarkers, as $\lrend z$. The computation starts at $q\start$ on the left endmarker. In each step, the next state and head move are derived from~$\ggd$ and the current state and symbol. Endmarkers are never violated, except for~$\rend$ when the next state is $q\accept$ or $q\reject$. Hence, the computation can either loop; or fall off~$\rend$ into one of $q\accept\kkk q\reject$.

Formally, the \textit{computation} of~$M$ from state~$p$ and ``the $j$th symbol of string~$z$'' (where $0\leq j\leq|z|\ppp1$), denoted by $\comp_{M,p,j}(z)$, is the unique sequence $c=(q_t, j_t)_{0\leq t<1+m}$ of state-position pairs, called \textit{configurations}, where $0\leq m\leq\infty$; $q_0:=p$ and $j_0:=j$; for all $0\leq t<m$, the $t$-th position satisfies $1\leq j_t\leq |z|$; and every next state~$q_{t+1}$ and position~$j_{t+1}$ are derived from~$\ggd$ and the current state~$q_t$ and symbol~$z_{j_t}$. If $m=\infty$, we say the computation \textit{loops}. Otherwise, $m$~is finite and the last position~$j_m$ is either $|z|\ppp1$ or~$0$, in which cases we say the computation \textit{hits right} or \textit{hits left} into~$q_m$, respectively.%
\footnote{Note that, if the starting position~$j$ is~$0$, then the computation~$c$ starts at the imaginary position just to the left of the leftmost boundary of $z$, so it immediately hits left (into~$p$). Likewise, if $j=|z|\ppp1$, then $c$~starts at the imaginary position just to the right of the rightmost boundary of $z$, so it immediately hits right into~$p$. In both cases, it is $m=0$ and $c=\bigl((q_0,j_0)\bigr)=\bigl((p,j)\bigr)$.\label{foot:computation-from-outside-input}} 
If $j=1$ or $j=|z|$, we get the \textit{left computation} and \textit{right computation} of~$M$ from~$p$ on~$z$, respectively:\footnotemark
\[
\lcomp_{M,p}(z):=\comp_{M,p,1}(z)
\qquad\text{and}\qquad
\rcomp_{M,p}(z):=\comp_{M,p, |z|}(z) \,.
\]
\footnotetext{Note that, if $z=\gge$ is the empty string, then every left computation on~$z$ hits right, since it starts (and ends immediately) at position $j=1=|z|\ppp1$. Likewise, every right computation on~$z$ hits left, since it starts (and ends immediately) at position $j=|z|=0$. See also Footnote~\ref{foot:computation-from-outside-input}.\label{foot:computation-on-empty-string}} 
Then the \textit{computation of~$M$ on $z$} is its left computation from~$q\start$ on $\lrend z$, in symbols 
\[
\comp_M(z):=\lcomp_{M,q\start}(\lrend z) \,;
\]
and $M$~\textit{accepts}~$z$ if this computation hits right into~$q\accept$. As usual, $M$~\textit{solves} a language $L\subs\ggS\s$ if it accepts exactly its strings: $M$~accepts~$z\iff z\in L$. 

\subsection{One-way liveness}
\label{sec:owl}

Let $h\geq1$. The alphabet $\ggS_h:=\mathcal{P}([h]\xxx[h])$ consists of every two-column directed graph with $h$~nodes per column and only rightward arrows (Fig.\,\ref{fig:owl}{\textsf{\scriptsize L}}). An $n$-long string $z\in\ggS\s_h$ is naturally viewed as a graph of $n\ppp1$~columns, indexed from~$0$ to~$n$, where edges connect successive columns and, for simplicity, are undirected (Fig.\,\ref{fig:owl}{\textsf{\scriptsize M}}). If there exists an $n$-long path from some node of column~$0$ to some node of column~$n$, then $z$~is \textit{live}; otherwise, it is \textit{dead}. Checking that a given $z\in\ggS_h\s$ is live is the \textit{one-way liveness} problem for height~$h$, denoted by~$\owl_h$. 

\begin{figure}[t]
\centering
\begin{tabular}{ccc}
\raisebox{0.00cm}{\includegraphics[height=2.4cm]{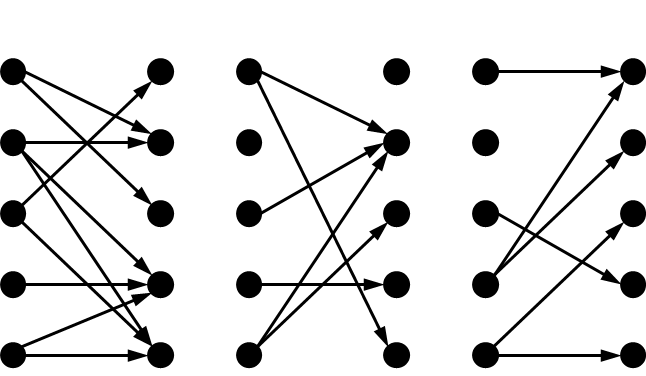}} & \qquad\quad
\raisebox{0.00cm}{\includegraphics[height=2.4cm]{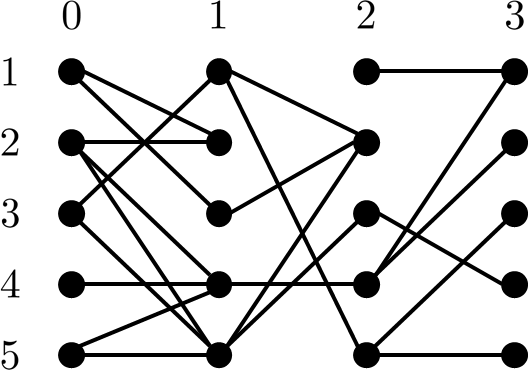}} & \qquad\quad
\raisebox{0.9cm}{
$\begin{bmatrix}
0 & 0 & 0 & 0 & 0\\
1 & 1 & 0 & 1 & 0\\
0 & 0 & 1 & 1 & 1\\
1 & 1 & 0 & 0 & 0\\ 
1 & 1 & 0 & 1 & 0\\ 
\end{bmatrix}$}
\end{tabular}
\caption{Three symbols in~$\ggS_5$~({\textsf{\scriptsize L}}); their string~({\textsf{\scriptsize M}}); and its connectivity~({\textsf{\scriptsize R}}).}
\label{fig:owl}
\end{figure}

Given any $n$-long $z\in\ggS_h\s$, its \textit{connectivity} $C(z)\in\bbBh$ is the $h \times h$ Boolean matrix where cell $(i,j)$ is~$1$ iff $z$ contains an $n$-long path from node~$i$ of column~$0$ to node~$j$ of column~$n$ (Fig.\,\ref{fig:owl}{\textsf{\scriptsize R}}). Easily, concatenation of strings corresponds to \textit{Boolean multiplication} of connectivities, namely $C(xy) = C(x)C(y)$, where every $+$ and~$\cdot$ on Boolean values is $\vee$ and~$\wedge$, respectively. So, equivalently, $\owl_h$ is defined over the alphabet $\bbBh$ and consists in checking that a given string of Boolean matrices of dimension $h\times h$ has non-zero Boolean product.

The problem is important in studying how easily \tdfas\ can simulate \textit{one-way nondeterministic finite automata} (\onfas, as typically defined~\cite{rasc59,si12}). Formally, the family $\owl=\fml\owl$ is complete, under homomorphic reductions, for the class \onen\ of all language families $L=\fml L$ that can be recognized by \onfa\ families $N=\fml N$ where the number of states grows polynomially in~$h$~\cite{sasi78}. Intuitively, $\owl_h$~captures in its instances all computations of all $h$-state \onfas. So, if a \tdfa\ can solve $\owl_h$ with $\atmost s(h)$~states, then every $h$-state \onfa\ can be simulated by a \tdfa\ with $\atmost s(h)$ states; and, conversely, if every \tdfa\ for $\owl_h$ needs $\atleast s(h)$~states, then some $h$-state \onfa\ requires $\atleast s(h)$ states in every \tdfa\ that simulates it.

\subsection{Properties}
\label{sec:properties}

A \textit{property} (of the strings) over an alphabet~$\ggS$ is any set $P\subs\ggS\s$. If $x\in P$, we also say that ``$x$~has~$P$''. We are interested in properties over $\ggS_h$ that are defined by connectivities: 

\begin{definition}
For any $C\in \bbBh$, we let $P(C):=\{ z\in\ggS_h\s \mid C(z)=C \}$ be the property of all strings over $\ggS_h$ that have connectivity~$C$.
\end{definition}

A property~$P$ over an alphabet~$\ggS$ is~\textit{smooth} if it is ``closed under concatenation with infix'', meaning that every two strings in the property have an infix over $\ggS$ that keeps their concatenation in the property, namely: $(\forall x\in P)(\forall z\in P)(\exists y\in\ggS\s)(xyz\in P)$.

\begin{example}\label{exa:smooth}
For any $i,j\in[h]$, let $C_{i,j}\in\bbBh$ be the matrix where cell $(i,j)$ is the only cell that contains~$1$. Then property $P_{i,j}:=P(C_{i,j})$ is smooth, because any two strings $x$ and~$z$ in it have the one-symbol infix $y:=\{(j,i)\}$ that ensures that $xyz$ is also in $P_{i,j}$.    
\end{example}

Given two properties, $P$ and~$P'$ over an alphabet $\ggS$, we say that $P$~\textit{\tagLR-extends} to~$P'$ if every string in~$P$ can be \tagLR-extended into a string in~$P'$, namely if $(\forall x\in P)(\exists y\in\ggS\s)(xy\in P')$.

\begin{example}\label{exa:LR-extends}
For any $i,j\in[h]$, let $P_{i,j}$ be as in Example~\ref{exa:smooth} and let $P_{i,\ast}:=\bigcup_k P_{i,k}$ consist of every string whose connectivity has~$1$ in some cell in row~$i$ and in no other cells. Then property $P_{i,\ast}$ \tagLR-extends to property $P_{i,j}$ because, for any string $x\in P_{i,\ast}$ whose connectivity has its only~$1$ in a cell~$(i,k)$, the one-symbol suffix $y:=\{(k,j)\}$ is such that $xy$ is in~$P_{i,j}$.    
\end{example}

Sometimes it is possible to find a single suffix~$y$ that can \tagLR-extend into~$P'$ every~$x$ from~$P$ and from~$P'$, namely $(\exists y\in\ggS\s)(\forall x\in P\cup P')(xy\in P')$. We then say that \textit{$P$~has a suffix into~$P'$}. 

\begin{example}\label{exa:suffix-into}
For any $i,j\in[h]$, let $P_{i,\ast}$ and $P_{i,j}$ be as in Example~\ref{exa:LR-extends}. Then $P_{i,\ast}$ (not only \tagLR-extends to $P_{i,j}$, but even) has a suffix into $P_{i,j}$. Indeed, the one-symbol suffix $y:=\{(k,j)\mid  k\in[h]\}$, that connects every node on the left to the $j$-th node of the right, forces $xy$ to be in~$P_{i,j}$, no matter what $x\in P_{i,\ast}\cup P_{i,j} = P_{i,\ast}$ we select.    
\end{example}

In even more special cases, the suffix that \tagLR-extends into $P'$ all strings of $P\cup P'$ can also be chosen to end in any way $v\in P$ that we want: $(\forall v\in P)(\exists u\in\ggS\s)(\forall x\in P\cup P')(xuv\in P')$. We then say that \textit{$P$~has suffix of choice into~$P'$}. 

\begin{example}\label{exa:suffix-of-choice}
Let $I,J\in\bbBh$ be the identity matrix and the matrix with $1$s in all of the first row and nowhere else, respectively. Then $P:=P(I)$ is the property of having $h$ disjoint paths that connect every leftmost node to the corresponding rightmost node; and $P':=P(I+J)$ is like $P$ except that the top leftmost node also connects to all rightmost nodes. We argue that $P$~has suffix of choice into $P'$. Indeed, pick any $v\in P$. Choose $u\in P'$ (any such $u$ will do) and consider the suffix $y:=uv$. Easily, for every $x\in P\cup P'$, appending $y$ forces $xy$ into~$P'$.
\end{example}

Symmetric variants of the above relations are also relevant: we say that \textit{$P$ \tagRL-extends to~$P'$} if $(\forall x\in P)(\exists y\in\ggS\s)(yx\in P')$; that \textit{$P$~has a prefix into~$P'$} if $(\exists y\in\ggS\s)(\forall x\in P\cup P')(yx\in P')$; and that \textit{$P$ has a prefix of choice into~$P'$} if $(\forall v\in P)(\exists u\in\ggS\s)(\forall x\in P\cup P')(vux\in P')$.

Finally, we say that two properties $P,P'\subs\ggS\s$ are \textit{separated} by a language~$L\subs\ggS\s$ if every two of their strings have a context where the choice between them affects membership in~$L$: 
\[
(\forall x\in P)(\forall z\in P')(\exists u\in\ggS\s)(\exists v\in\ggS\s)(uxv\in L \iff uzv\not\in L) \,.
\]
In particular, distinct connectivities define properties that are separated by~\owl:

\begin{lemma}\label{lem:distinct-conns-separated}
Let $C,C'\in\bbBh$. If $C\neq C'$ then $P(C),P(C')$ are separated by~$\owl_h$.
\end{lemma}

\begin{proof}
Since $C\neq C'$, there exists cell $(i,j)$ such that $C(i,j)\neq C'(i,j)$. Without loss of generality, assume $C(i,j)=0$ and $C'(i,j)=1$. Now pick any $x\in P(C)$ and $z\in P(C')$. Then $x$~has no path from the $i$th leftmost node to the $j$th rightmost node, but $z$~does. Consider the symbols $u:=\{(1,i)\},v:=\{(j,1)\}\in\ggS_h$. Clearly, $uxv$ is dead and $uzv$ is live.
\end{proof}

\section{A lower-bound lemma}
\label{sec:the-lower-bound-lemma}

Let $L$ be a language over some alphabet $\ggS$. This section proves a general lemma for lower-bounding the number of states in a \tdfa\ solving~$L$. Section~\ref{sec:the-lower-bound} uses this lemma for $\ggS=\ggS_h$ and $L=\owl_h$ to show a lower bound of~$\ggV(h^2)$.

So, fix a \tdfa\ $M$ that solves~$L$ with states~$Q$. (Since $M$ is fixed, we drop unnecessary subscripts: e.g., for $\lcomp_{M,p}(z)$, we just write $\lcomp_p(z)$.)

\subsection{Generic strings}
\label{sec:generic-strings}

Introduced in~\cite{si80a}, generic strings continue to be the main building block that we have for constructing hard inputs for \tdfas\ of various kinds~\cite{ka06icalp,ka07jalc,kakrmo12jcss,ka13iacc,kapi12mfcs}. Here, we are using the flavor of~\cite[Sect.~3.2]{ka07jalc}. This subsection briefly recalls concepts and facts from that work that we need in our arguments.

For $y\in\ggS\s$, consider the set of ``\textit{exit states}'' that $M$~hits right into at the end of a left computation on~$y$:
\[
Q\subLR(y) := \{ q\in Q \mid (\exists p\in Q)(\lcomp_p(y)\text{ hits right into }q) \} \,.
\]
When we \tagLR-extend~$y$ to some~$yz$, the size of this set cannot possibly increase, because the function $\gga_{y,z}$ that maps every $q\in Q\subLR(y)$ to the state (if any) that is hit right by the computation $\comp_{q,|y|+1}(yz)$ is easily seen to be a partial surjection from $Q\subLR(y)$ onto~$Q\subLR(yz)$.  

\begin{lemma}[{\cite[Fact~3]{ka07jalc}}]\label{lem:ayz-is-partial-surjection}
For all $y\kkk z$, the function $\gga_{y,z}$ partially surjects $Q\subLR(y)$ onto~$Q\subLR(yz)$; hence $|Q\subLR(y)|\geq |Q\subLR(yz)|$.  
\end{lemma}

Consequently, if we start with a string~$y$ in a property~$P$ and keep \tagLR-extending it inside~$P$, we will at some point have minimized the size of the set of exit states. When that happens, the string has become ``\tagLR-generic''. 

\begin{definition}[{\cite[Def.~3]{ka07jalc}}]
A string $y$ is \textit{\tagLR-generic} for property~$P$ if $y\in P$ and, for all $yz\in P$, we have $|Q\subLR(y)| = |Q\subLR(yz)|$.
\end{definition}

\begin{lemma}[{\cite[Lemma~3]{ka07jalc}}]\label{lem:LR-generics-exist}
Every property $P\neq\emptyset$ has \tagLR-generic strings.
\end{lemma}

\begin{proof}
Pick $y\in P$. If $y$ is \tagLR-generic, we are done. Otherwise, there is $yz_1\in P$ such that $|Q\subLR(y)|>|Q\subLR(yz_1)|$. So, we replace $y$ with $yz_1$ and repeat: either $yz_1$ is \tagLR-generic and we are done; or there is $yz_1z_2\in P$ such that $|Q\subLR(yz_1z_2)|$ is even smaller, and we replace $yz_1$ with $yz_1z_2$ and repeat. Clearly, this eventually terminates, as $M$ has finitely many states and every step decreases $|Q\subLR(\something)|$.
\end{proof}

In addition, every further \tagLR-extension of a \tagLR-generic string inside the property continues to be \tagLR-generic for the property: 

\begin{lemma}\label{lem:LR-extension-preserves-genericity}
If $y$ is \tagLR-generic for~$P$, then so is every $yz\in P$.
\end{lemma}

\begin{proof}
Towards a contradiction, say $y$ is \tagLR-generic for~$P$, but some $yz\in P$ is not. Then $|Q\subLR(y)|=|Q\subLR(yz)|$ (since $y$~is \tagLR-generic and $yz\in P$) but there exists $yzz'\in P$ such that $|Q\subLR(yz)|>|Q\subLR(yzz')|$ (since $yz$ is not \tagLR-generic). Then $\tilde z:=zz'$ is such that $y\tilde z\in P$ but $|Q\subLR(y)|>|Q\subLR(y\tilde z)|$, contrary to the \tagLR-genericity of~$y$.
\end{proof}

Finally, the size $|Q\subLR(yz)|$ is upper bounded, not only by the respective size $|Q\subLR(y)|$ for the prefix~$y$ (the ``\textit{prefix-monotonicity}'' given by Lemma~\ref{lem:ayz-is-partial-surjection}), but also by the respective size $|Q\subLR(z)|$ for the suffix~$z$ (a ``\textit{suffix-monotonicity}'').  

\begin{lemma}[{\cite[Fact\,4]{ka07jalc}}]\label{lem:suffix-mononotonicity}
\!For all $y\kkk z$: $Q\subLR(yz)\subs Q\subLR(z)$; so $|Q\subLR(yz)|\leq |Q\subLR(z)|$.
\end{lemma}

Symmetrically to the above, we can also define the set $Q\subRL(y)$ of \textit{exit states} that $M$~hits left into after right computations on~$y$: 
\[
Q\subRL(y) := \{ q\in Q \mid (\exists p\in Q)(\rcomp_p(y)\text{ hits left into }q) \} \,;
\]
then introduce the partial surjection $\ggb_{z,y}:Q\subRL(y)\to Q\subRL(zy)$ that maps every~$q$ to the state (if any) that is hit left by $\comp_{q,|z|}(zy)$; and use it to conclude that $|Q\subRL(y)|$ cannot increase when we \tagRL-extend~$y$ (\textit{suffix-monotonicity}); then define \textit{\tagRL-generic strings}:
\begin{definition}[{\cite[Def.~3]{ka07jalc}}]
A string $y$ is \textit{\tagRL-generic} for property~$P$ if $y\in P$ and, for all $zy\in P$, we have $|Q\subRL(y)| = |Q\subRL(zy)|$.
\end{definition}
\noindent
and show that they exist for every $P\neq\emptyset$; then show that \tagRL-extensions preserve \tagRL-genericity; and finally argue that $Q\subRL(z)\sups Q\subRL(zy)$ and therefore \textit{prefix-monotonicity} also holds.

In the end, we are interested in strings that are generic in both directions, and we are guaranteed to have them for non-empty properties that are also smooth. 

\begin{definition}[{\cite[Def.~3]{ka07jalc}}]
A string is \textit{generic} if it is both \tagLR- and \tagRL-generic.
\end{definition}

\begin{lemma}[{\cite[Lemma\,3]{ka07jalc}}]\label{lem:generics-exist}
Every smooth property $P\neq\emptyset$ has generic strings.
\end{lemma}

\begin{proof}
Since $P\neq\emptyset$, there exist $x,z\in P$ that are respectively \tagLR- and \tagRL-generic for~$P$ (by Lemma~\ref{lem:LR-generics-exist} and its symmetric variant). Since $P$~is smooth, there exists~$y$ such that $w:=xyz\in P$. Then $w$~is both \tagLR- and \tagRL-generic for~$P$, as a \tagLR-extension of~$x$ and a \tagRL-extension of~$z$ (Lemma~\ref{lem:LR-extension-preserves-genericity} and its symmetric variant). Therefore, $w$~is generic for~$P$.
\end{proof}

\subsection{Exit sizes}
\label{sec:exit-sizes}

The previous section studied how the size of the set $Q\subLR(y)$ changes as $y\in P$ is \tagLR-extended within~$P$. Call this number $a(y):=|Q\subLR(y)|$ the ``\textit{\tagLR-size}'' of~$y$. We now study how this number changes when $y$~is replaced by other strings. 

We first note that switching between \tagLR-generic strings for~$P$ preserves the \tagLR-size, provided that $P$ is smooth. Hence, this number is a feature of~$P$.  

\begin{lemma}\label{lem:LR-generics-have-same-LRsize}
Let $x\kkk z$ be \tagLR-generic strings for a smooth~$P$. Then $a(x)=a(z)$.
\end{lemma}

\begin{proof}
Suppose $P$~is smooth and strings $x,z$ are \tagLR-generic for it. Since $P$ is smooth, there exists $y$ such that $xyz\in P$. Then $a(x)=|Q\subLR(x)|=|Q\subLR(xyz)|$, by the definition of $a(x)$ and since $xyz\in P$ \tagLR-extends the \tagLR-generic~$x$. But also $|Q\subLR(xyz)|\leq|Q\subLR(z)|=a(z)$, by suffix monotonicity (Lemma~\ref{lem:suffix-mononotonicity}) and the definition of $a(z)$. Overall, $a(x)\leq a(z)$. Finally, $a(z)\leq a(x)$ is also true, symmetrically.
\end{proof}

\begin{definition}\label{def:LR-size}
The \emph{\tagLR-size} $a(P)$ of a smooth~$P\neq\emptyset$ is the common \tagLR-size of all \tagLR-generic strings of~$P$. 
\end{definition}

Moreover, $a(P)$~is a lower bound for the \tagLR-size of any string in~$P$, and is reached by exactly the \tagLR-generic strings of~$P$.  

\begin{lemma}\label{lem:LR-size-characterizes-LR-generics}
Let $P\neq\emptyset$ be smooth. Every $x\in P$ satisfies $a(x)\geq a(P)$; and the equality holds iff $x$ is \tagLR-generic for~$P$.
\end{lemma}

\begin{proof}
Let $x\in P$. As in the proof of Lemma~\ref{lem:LR-generics-exist}, $x$~can be \tagLR-extended to some \tagLR-generic string $xz$ for~$P$. Then $a(x)=|Q\subLR(x)|\geq|Q\subLR(xz)|=a(xz)=a(P)$ (by Lemma~\ref{lem:ayz-is-partial-surjection} and the definitions of $a(\something),a(P)$), hence $a(x)\geq a(P)$.

Now let~$x$ be \tagLR-generic for~$P$. Then $a(x)=a(P)$ (Def.~\ref{def:LR-size}). Conversely, assume $a(x)=a(P)$. Pick any \tagLR-extension $xz\in P$. We know $|Q\subLR(x)|\geq|Q\subLR(xz)|$, by Lemma~\ref{lem:ayz-is-partial-surjection}. But we also know that $|Q\subLR(x)|\leq|Q\subLR(xz)|$, because  $|Q\subLR(x)|=a(x)=a(P)$ (by definition of $a(x)$ and our assumption) and $a(P)\leq|Q\subLR(xz)|$ (by the first part of our proof, and since $xz\in P$). Overall, $|Q\subLR(x)|=|Q\subLR(xz)|$ for the arbitrary \tagLR-extension of $x$ in~$P$. So, $x$~is \tagLR-generic for~$P$.
\end{proof}

Finally, the \tagLR-sizes of distinct smooth properties can be compared, if one \tagLR-extends to the other:

\begin{lemma}\label{lem:LR-extension-affects-property-LR-size}
If $P,P'\!\neq\emptyset$ are smooth and $P$ \tagLR-extends to $P'$\!\!, then $a(P)\geq a(P')$.
\end{lemma}

\begin{proof}
Pick any \tagLR-generic string~$x$ for $P$. Since $P$ \tagLR-extends to $P'$, there exists $y$ such that $xy\in P'$. 
Then $a(P)=a(x)=|Q\subLR(x)|\geq|Q\subLR(xy)|\geq a(P')$, where the two equalities are by the definitions of $a(P)$ and $a(x)$ (and since $x$~is \tagLR-generic) and the two inequalities hold by Lemmas~\ref{lem:ayz-is-partial-surjection} and~\ref{lem:LR-size-characterizes-LR-generics}, respectively.  
\end{proof}

Symmetrically to the above, we can also define the \textit{\tagRL-size} $b(y):=|Q\subRL(y)|$ of a string~$y$; prove that it is common for all \tagRL-generic strings of a smooth~$P$; and define it as the \textit{\tagRL-size} $b(P)$ of~$P$:
\begin{definition}\label{def:RL-size}
The \emph{\tagRL-size} $b(P)$ of a smooth~$P\neq\emptyset$ is the common \tagRL-size of all \tagRL-generic strings of~$P$. 
\end{definition}
\noindent
then show that it lower bounds the \tagRL-size of all $x\in P$ and is met iff $x$~is \tagRL-generic; and prove $b(P')\leq b(P)$ when a smooth~$P$ \tagRL-extends to a smooth~$P'$. We then also prove the following.

\begin{lemma}\label{lem:suffix-of-choice-affects-property-exit-sizes}
If $P,P'\neq\emptyset$ are smooth and $P$ has suffix of choice into~$P'$, then both $a(P)\geq a(P')$ and\/ $b(P)\geq b(P')$.
\end{lemma}

\begin{proof}
Since $P$~has suffix of choice into~$P'$, it clearly \tagLR-extends to~$P'$, therefore $a(P)\geq a(P')$ (by Lemma~\ref{lem:LR-extension-affects-property-LR-size}), and we just need to prove that $b(P)\geq b(P')$.

Let $v$~be any generic string for~$P$ (Lemma~\ref{lem:generics-exist}). Then $v\in P$ and, since $P$~has suffix of choice into~$P'$, some $u$~is such that the suffix~$uv$ forces a string $xuv$ into~$P'$, no matter which $x\in P\cup P'$ we pick. Let us pick $x=v$. Then $vuv\in P'$, and it suffices to prove that $b(P)\geq b(vuv)\geq b(P')$. 

The second inequality follows by the symmetric variant of~Lemma~\ref{lem:LR-size-characterizes-LR-generics} and since $vuv\in P'$. The first one follows from the fact that $b(P)=b(v)=|Q\subRL(v)|\geq|Q\subRL(vuv)|=b(vuv)$, where the three equalities are by the definitions of $b(P)$ (since $v$~is \tagRL-generic for~$P$), $b(v)$, and $b(vuv)$, and the inequality is by prefix-monotonicity (the symmetric variant of Lemma~\ref{lem:suffix-mononotonicity}).
\end{proof}

\subsection{The lemma}
\label{sec:the-lemma}

We are now ready to state and prove our main lemma.

\begin{lemma}[Main Lemma]\label{lem:main}
Suppose language $L\subs\ggS\s$ is solved by \tdfa~$M$. Let\/ $\emptyset\neq P,P'\subs\ggS\s$ be two smooth properties separated by~$L$. If $P$~has suffix of choice into~$P'$, then the exit sizes of $P,P'$ on $M$ satisfy $a(P)\geq a(P')$ and $b(P)\geq b(P')$, and at least one of these inequalities is strict.
\end{lemma}

\begin{proof}
Inequalities $a(P)\geq a(P')$ and $b(P)\geq b(P')$ follow from Lemma~\ref{lem:suffix-of-choice-affects-property-exit-sizes}. So we just need to prove that at least one of the two is strict. Toward a contradiction, suppose that neither is, namely $a(P)=a(P')$ and $b(P)=b(P')$.

Let $\ggu$ be generic for $P$ (Lemma~\ref{lem:generics-exist}). Since $P$~has suffix of choice into~$P'$, there exists~$x$ such that the suffix~$x\ggu$ forces every string $wx\ggu$ into~$P'$, for all $w\in P\cup P'$:
\begin{equation}\label{equ:xθ-is-suffix-into-P'}
    (\forall w\in P\cup P')(wx\ggu\in P')
\end{equation}
\begin{claim}\label{clm:all-θxθ...xθ-are-in-P'}
For all $t\geq1$, we have $\ggu(x\ggu)^t\in P'$.    
\end{claim}

\begin{proof}
By induction on~$t\geq1$. For $t=1$, apply~\eqref{equ:xθ-is-suffix-into-P'} for $w:=\ggu$ (which is in~$P$ as generic string), to conclude that $P'$~contains $w x\ggu=\ggu x\ggu$. For $t\geq2$, apply~\eqref{equ:xθ-is-suffix-into-P'} for $w:=\ggu(x\ggu)^{t-1}$ (which is in~$P'$ by the inductive hypothesis), to conclude that $P'$ also contains $w x\ggu=\ggu(x\ggu)^{t-1} x\ggu=\ggu(x\ggu)^t$.
\end{proof}

Now, since $P,P'$ are separated by~$L$, we know that, for every $t\geq1$, there exists a context $u_t,v_t$ such that exactly one of $u_t\ggu v_t$ and $u_t\ggu(x\ggu)^tv_t$ is in~$L$ (since $\ggu\in P$ and $\ggu(x\ggu)^t\in P'$). In what follows, we find~$t\geq1$ such that $M$~decides identically on $u_t\ggu v_t$ and $u_t\ggu(x\ggu)^t v_t$, contradicting the assumption that $M$~solves~$L$.

We first focus on the string $\ggu(x\ggu)^t$ for the arbitrary $t\geq1$ and on the function $\gga_{\ggu,(x\ggu)^t}$ that partially surjects $Q\subLR(\ggu)$ to $Q\subLR(\ggu(x\ggu)^t)$ (Lemma~\ref{lem:ayz-is-partial-surjection}). Letting $A:=Q\subLR(\ggu)$, we prove:

\begin{landscape}
\begin{figure}
\raisebox{1.8cm}{(a)}\quad\hspace{0.75pt}%
\raisebox{0.555cm}{\includegraphics[height=4cm]{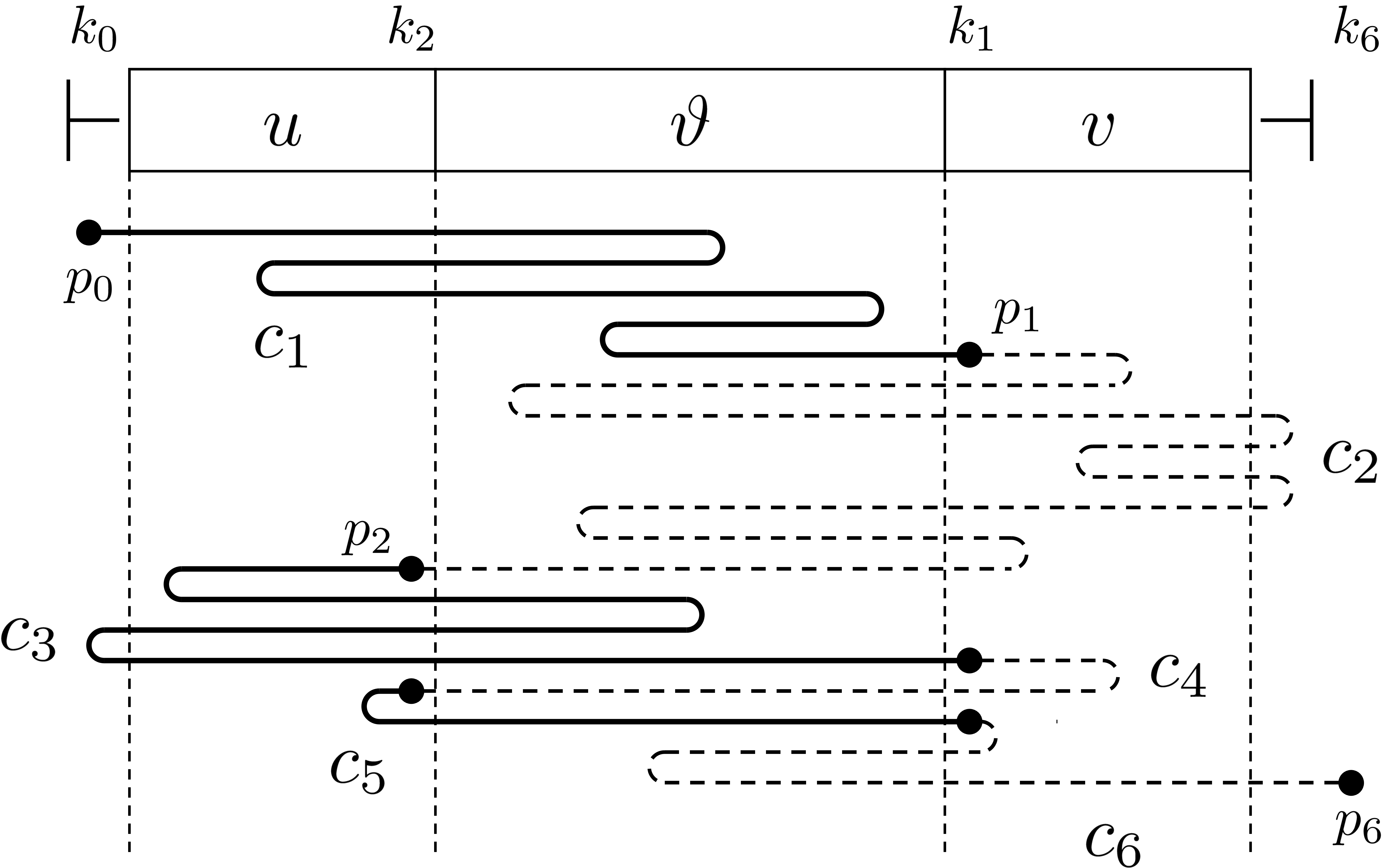}}
\qquad\hspace{1em}
\includegraphics[height=4.25cm]{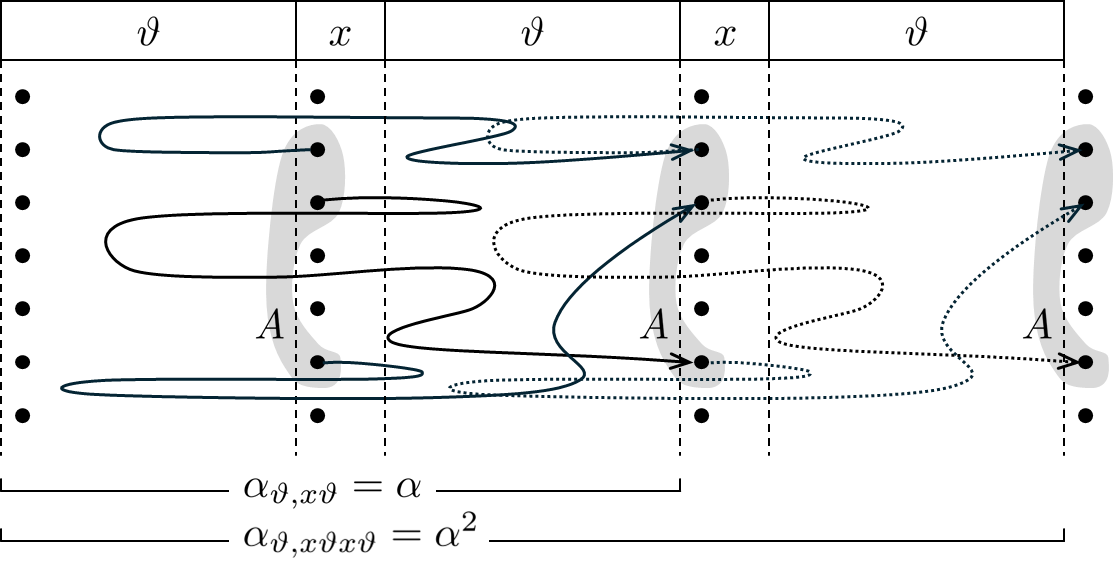}
\quad\raisebox{1.8cm}{(c)}
\par\bigskip
\raisebox{2.3cm}{(b)}\quad
\includegraphics[height=5.25cm]{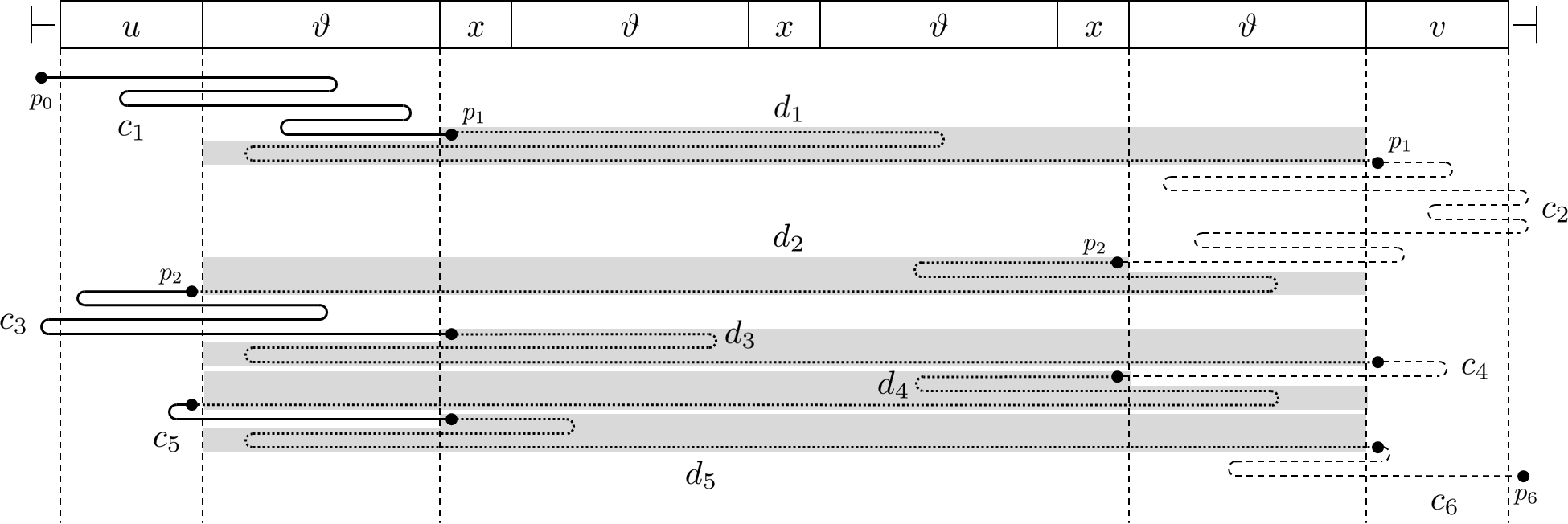}
\caption{(a)~A computation on $u\ggu v$ with $\ell+2=7$ crucial points; and the corresponding $\ell+1=6$ computation infixes $c_1,\dots,c_6$, odd (bold) and even (dashed). %
(b)~The computation $c'$ on $u\ggu(x\ggu)^tv$ (here $t=3$), created by joining the concatenations $c_i'$ of the~$c_i$ and~$d_i$. %
(c)~How $\gga:=\gga_{\ggu,x\ggu}$ (solid lines) permutes $A:=Q\subLR(\ggu)$; and $\gga_{\ggu,(x\ggu)^t}$ (solid \&\ dotted) permutes $A$, too, sometimes as identity (here $t=2)$.}
\label{fig:computations}
\end{figure}
\end{landscape}

\begin{claim}\label{clm:every-α(t)-permutes-A}
For all $t\geq1$, the function $\gga_{\ggu,(x\ggu)^t}$ is a permutation of~$A$.
\end{claim}

\begin{proof}
We know $Q\subLR(\ggu (x\ggu)^t)$ is a subset of $Q\subLR(\ggu)$ (by Lemma~\ref{lem:suffix-mononotonicity} and since $\ggu$~is a suffix of~$\ggu(x\ggu)^t$); and at least as big as $Q\subLR(\ggu)$, because $|Q\subLR(\ggu(x\ggu)^t)|=a(\ggu(x\ggu)^t)\geq a(P')$ (by Lemma~\ref{lem:LR-size-characterizes-LR-generics} and since $\ggu(x\ggu)^t\in P'$) and $a(P')=a(P)$ (our assumption) and $a(P)=|Q\subLR(\ggu)|$ (by definition of $a(P)$ and since $\ggu$~is \tagLR-generic). So, it must be $Q\subLR(\ggu(x\ggu)^t)=Q\subLR(\ggu)=A$, which implies that $\gga_{\ggu,(x\ggu)^t}$ is total and injective. Since it is also surjective, it bijects~$A$ to~$A$. (Fig.~\ref{fig:computations}c.)
\end{proof}

Moreover, each one of these permutations $\gga_{\ggu,(x\ggu)^t}$ for general $t\geq1$ is nothing more than the $t$-th iteration of the very first permutation $\gga:=\gga_{\ggu,x\ggu}$ for $t=1$ (Fig.~\ref{fig:computations}c):

\begin{claim}\label{clm:α(t)=α^t}
For all $t\geq1$: \; $\gga_{\ggu,(x\ggu)^t}=\gga^t$.    
\end{claim}

\begin{proof}
By induction on $t\geq1$. In the base case ($t=1$), we clearly have 
$
\gga_{\ggu,(x\ggu)^1}=
\gga_{\ggu,x\ggu}=
\gga=
\gga^1
$, 
by the definitions. In the inductive step ($t\geq1$), assume $\gga_{\ggu,(x\ggu)^t}=\gga^t$ and consider the permutation $\gga_{\ggu,(x\ggu)^{t+1}}$ that maps every $q\in Q\tagLR(\ggu)=A$ to the state in $Q\subLR(\ggu(x\ggu)^{t+1})=Q\subLR(\ggu)=A$ that is hit right by the computation 
\[
c:=\comp_{q,|\ggu|+1}\bigl(\; \ggu(x\ggu)^{t+1} \;\bigr)\,. 
\]
We prove that $\gga_{\ggu,(x\ggu)^{t+1}}$ is simply the composition $\gga_{\ggu,x\ggu}\circ\gga_{\ggu,(x\ggu)^t}$, namely:
\begin{equation}\label{equ:composition}
\text{for all $q\in Q\tagLR(\ggu)$}:\qquad
\gga_{\ggu,(x\ggu)^{t+1}}(q) = \gga_{\ggu,(x\ggu)^t}\bigl(\; \gga_{\ggu,x\ggu}(q) \;\bigr)
\end{equation}
and thus (by the definitions and the inductive hypothesis) the composition $\gga\circ\gga^t=\gga^{t+1}$, as desired.  

To prove \eqref{equ:composition}, fix any $q\in Q\tagLR(\ggu)$. Since $\gga_{\ggu,(x\ggu)^{t+1}}$~is a permutation (Claim~\ref{clm:every-α(t)-permutes-A}), we know $\gga_{\ggu,(x\ggu)^{t+1}}(q)$~is defined, so computation~$c$ hits right into some state~$s$ and $\gga_{\ggu,(x\ggu)^{t+1}}(q)=s$. Since it hits right, $c$~surely crosses the $\ggu x\ggu$-$(x\ggu)^t$ boundary; let~$r$ be the state produced by the first such crossing. Then $c$~starts with the prefix $c_1:=\comp_{q,|\ggu|+1}(\ggu x\ggu)$, which hits right into~$r$, causing $\gga_{\ggu,x\ggu}(q)=r$; and continues with the suffix $c_2:=\comp_{r,|\ggu x\ggu|+1}(\ggu(x\ggu)^{t+1})$, which hits right into~$s$. In turn, this suffix starts with the prefix $c_2':=\comp_{r,|\ggu|+1}(\ggu(x\ggu)^t)$, which is guaranteed to hit right (because $r\in Q\tagLR(\ggu)=A$ and, by Claim~\ref{clm:every-α(t)-permutes-A}, $\gga_{\ggu,(x\ggu)^t}$ is defined on all of~$A$), into state~$\gga_{\ggu,(x\ggu)^t}(r)$. But then $c_2'$~is all of~$c_2$, namely $c_2'=c_2$; and thus the two computations hit right into the same state, namely $\gga_{\ggu,(x\ggu)^t}(r)=s$. Putting everything together: 
\[
\gga_{\ggu,(x\ggu)^{t+1}}(q) = 
s =
\gga_{\ggu,(x\ggu)^t}(r) = 
\gga_{\ggu,(x\ggu)^t}\bigl(\; \gga_{\ggu,x\ggu}(q) \;\bigr) \,
\]
as required. 
\end{proof}

Finally, by symmetric arguments, the partial surjection $\ggb:=\ggb_{\ggu x,\ggu}$ (cf.~comments after Lemma~\ref{lem:suffix-mononotonicity}) of $B:=Q\subRL(\ggu)$ to~$Q\subRL(\ggu x\ggu)$ is also a permutation of~$B$; and for every $t\geq1$, the respective mapping $\ggb_{(\ggu x)^t,\ggu}$ is the permutation $\ggb^t$ of~$B$.

Now, let $t\subLR,t\subRL\geq1$ be such that $\gga^{t\subLR}$ and $\ggb^{t\subRL}$ become the identity maps $\ID_A$ on~$A$ and $\ID_B$ on~$B$, respectively. Then for $t=t\z:=t\subLR\cdot t\subRL$ and for the string $\ggu(x\ggu)^t$, we have $\gga_{\ggu,(x\ggu)^t}=\gga^t=(\gga^{t\subLR})^{t\subRL}=(\ID_A)^{t\subRL}=\ID_A$ and likewise $\ggb_{(\ggu x)^t,\ggu}=\ggb^t=\ID_B$. Intuitively, this means that no \tagLR-traversal of $\ggu(x\ggu)^t$ by our \tdfa~$M$ ``notices'' the $t$~copies of $x\ggu$ to the right of the leftmost~$\ggu$; neither any \tagRL-traversal ``notices'' the $t$~copies of~$\ggu x$ to the left of the rightmost~$\ggu$. This becomes an issue, when we place the two strings $\ggu$ and $\ggu(x\ggu)^t$ in the context of $u:=u_t$ and $v:=v_t$.

\begin{claim}
$M$~decides identically on $u\ggu v$ and $u\ggu(x\ggu)^t v$.
\end{claim}

\begin{proof}
Consider the computation $c$ of $M$ on $u\ggu v$ (Fig.~\ref{fig:computations}a). Let $0\leq\ell\leq\infty$ be the number of times $c$~fully traverses~$\ggu$ (i.e., crosses into~$\ggu$ from one side and, after computing strictly inside~$\ggu$, crosses out of~$\ggu$ on the opposite side).%
\footnote{Note that, if $\ggu=\gge$ is the empty string, then ``$c$~fully traverses~$\ggu$'' reduces to just ``$c$~crosses the single outer-boundary of~$\ggu$''; namely, in this case (where $u\ggu v=uv$), to just ``$c$~crosses the $u$-$v$ boundary''. See also Footnotes~\ref{foot:computation-from-outside-input} and~\ref{foot:computation-on-empty-string}.}
Note that, either $\ell=\infty$ and $c$~loops in a way that fully traverses~$\ggu$ infinitely often; or $\ell<\infty$ and, after the $\ell$th traversal, $c$~ends with a suffix that either lies entirely inside $\lend u\ggu$ and is infinite, if $\ell$~is even; or lies entirely inside $\ggu v\rend$ (except possibly for the last configuration, if it exists) and is infinite or finite, if $\ell$~is odd.

Now, for $1\leq i<\ell\ppp1$, let the $i$th \textit{crucial point} of~$c$ be the configuration $(p_i,k_i)$ at the end of the $i$th full traversal, so that $k_i$ is either $|\lend u\ggu|\ppp1$ (odd~$i$) or~$|\lend u|$ (even~$i$); let the $0$th crucial point $(p_0,k_0)$ be the start configuration $(q\start,1)$; and, in case $c$~halts, let the $\ell\ppp1$st crucial point $(p_{\ell+1},k_{\ell+1})$ be the last configuration, off~$\rend$. Then call~$c_i$ the infix of~$c$ from the $i-1$st to the $i$th crucial point (inclusive), if $1\leq i<\ell\ppp1$; or the suffix of~$c$ from the $\ell$th crucial point onward, if $\ell<\infty$ and~$i=\ell\ppp1$. Note that every~$c_i$ lies entirely inside either $\lend u\ggu$ (if $i$~is odd) or $\ggu v\rend$ (if $i$~is even), except only for its last configuration (when it exists).

Now, for $1\leq i<\ell\ppp1$, let~$d_i$ be the computation of~$M$ that justifies why $\gga_{\ggu,(x\ggu)^t}(p_i)=\gga^t(p_i)=\ID_A(p_i)=p_i$, if $i$~is odd; or why $\ggb_{(\ggu x)^t,\ggu}(p_i)=\ggb^t(p_i)=\ID_B(p_i)=p_i$, if $i$~is even. Form~$c_i'$ by ``joining'' $c_i$ and $d_i$ at the shared configuration $(p_i,k_i)$, if $i$~is odd; or by ``joining'' a $\ggl$-shifted variant of~$c_i$ and $d_i$ at the shared configuration $(p_i,\ggl+k_i)$, where $\ggl:=|(x\ggu)^t|=t\cdot|x\ggu|$, if $i$~is even (Fig.~\ref{fig:computations}b). If $\ell<\infty$ and $i=\ell\ppp1$, then also let $c_i'$ be~$c_i$, if $i$~is odd; or a $\ggl$-shifted variant of~$c_i$, if $i$~is even. Then ``join'' all $c_i'$ to obtain the full computation $c'$ of~$M$ on $u\ggu(x\ggu)^tv$. 

By construction, $c'$~loops or halts exactly as~$c$: either because both computations loop by oscillating between $u$ and~$v$ infinitely often, if $\ell=\infty$; or because they both end by looping inside~$\lend u\ggu$, if $\ell<\infty$ is even; or because they both end by looping inside~$\ggu v\rend$, if $\ell<\infty$ is odd and $c_{\ell+1}$~is infinite; or because they both end by falling off~$\rend$ into the same state $q\accept$ or $q\reject$, if $\ell<\infty$ is odd and $c_{\ell+1}$~is finite. Therefore, $M$~decides identically on $u\ggu v$ and $u\ggu(x\ggu)^tv$.
\end{proof}

Since exactly one of  $u\ggu v$, $u\ggu(x\ggu)^t v$ is in~$L$, we know $M$~fails to solve~$L$ ---a contradiction.
\end{proof}

\section{The lower bound}
\label{sec:the-lower-bound}

To prove the promised lower bound, we introduce a sequence of $1+N$~non-empty properties
\begin{equation}\label{equ:P0,P1,...PN}
\emptyset\neq P_0,P_1,\dots,P_N\subs \ggS_h\s    
\end{equation}
such that $N=\tbinom{h+1}{2}$ and every two successive ones, $P_{t-1},P_t$ ($1\leq t\leq N$), satisfy the conditions of the Main Lemma: they are smooth (Lemma~\ref{lem:every-Pt-is-smooth}), separated by~$\owl_h$ (Lemma~\ref{lem:successive-Pi-separated-by-OWL}), and such that the earlier one has suffix of choice into the later one (Lemma~\ref{lem:Pt-1-has-suffix-of-choice-into-Pt}). Then the lemma implies that the exit sizes of~$P_{t-1}$ upper bound those of~$P_t$, and strictly so on at least one of the two sides. Therefore, in the sequence of the corresponding pairs of exit sizes
\[
(a(P_0),b(P_0)),\;
(a(P_1),b(P_1)),\;
\dots,\;
(a(P_N),b(P_N)),
\]
no step increases a component; and each step decreases at least one. This is a total of at least $N$ decrements of component, to be divided between the two nonincreasing sequences  
\[
a(P_0),
a(P_1),
\dots,
a(P_N)
\qquad\text{and}\qquad
b(P_0),
b(P_1),
\dots,
b(P_N) \,.
\]
So, at least one of them decreases at least $N/2$~times. Since all numbers are obviously between $|Q|$ and~$0$, it follows that $|Q|\geq N/2=\tfrac{1}{2}\tbinom{h+1}{2}=\frac14(h+1)h=\ggV(h^2)$.

To define the properties of~\eqref{equ:P0,P1,...PN}, we first define a sequence of connectivities $C_0,C_1\dots,C_N$ and then let each~$P_t$ be the property of having connectivity~$C_t$, namely $P_t:=P(C_t)$. So, in the rest of this section, we introduce the connectivities; prove a list of basic facts about them (Lemmas~\ref{lem:properties-of-Ct-with-Et}-\ref{lem:properties-of-Ct}); then use these to prove the desired facts about our properties (Lemmas~\ref{lem:every-Pt-is-smooth}-\ref{lem:Pt-1-has-suffix-of-choice-into-Pt}) and thus our main result (Theorem~\ref{thm:main}).

\subsection{The connectivities}
\label{sec:connectivities}

To construct the connectivities $C_0,C_1,\dots,C_N$, we take the identity matrix~$\IDTYh$ and start flipping its $h^2-h$ non-diagonal bits, from $0$ to~$1$, starting from the strict upper triangle (stage~1) and continuing to the strict lower triangle (stage~2), until all bits are~$1$. In the end (stage~3), we flip all $h^2$~bits together, from~$1$ back to $0$, to get the zero matrix~$\ZEROh$.

\begin{figure}[t] \centering
$\;C_0{=}
\begin{bmatrix}
1 & 0 & 0 & 0 & 0\\
0 & 1 & 0 & 0 & 0\\
0 & 0 & 1 & 0 & 0\\
0 & 0 & 0 & 1 & 0\\ 
0 & 0 & 0 & 0 & 1\\ 
\end{bmatrix}$
$\;C_6{=}
\begin{bmatrix}
1 & 0 & 0 & 0 & 1\\
0 & 1 & 0 & \underline1 & 1\\
0 & 0 & 1 & 1 & 1\\
0 & 0 & 0 & 1 & 1\\
0 & 0 & 0 & 0 & 1
\end{bmatrix}$
$C_{10}{=}
\begin{bmatrix}
1 & \underline1 & 1 & 1 & 1\\
0 & 1 & 1 & 1 & 1\\
0 & 0 & 1 & 1 & 1\\
0 & 0 & 0 & 1 & 1\\
0 & 0 & 0 & 0 & 1
\end{bmatrix}$
\\[2ex]
$C_{11}{=}
\begin{bmatrix}
1 & 1 & 1 & 1 & 1\\
0 & 1 & 1 & 1 & 1\\
0 & 0 & 1 & 1 & 1\\
0 & 0 & 0 & 1 & 1\\
0 & 0 & 0 & \underline1 & 1
\end{bmatrix}$
$C_{12}{=}
\begin{bmatrix}
1 & 1 & 1 & 1 & 1\\
0 & 1 & 1 & 1 & 1\\
0 & 0 & 1 & 1 & 1\\
0 & 0 & \underline1 & 1 & 1\\
0 & 0 & \underline1 & 1 & 1
\end{bmatrix}$
$C_{14}{=}
\begin{bmatrix}
1 & 1 & 1 & 1 & 1\\
\underline1 & 1 & 1 & 1 & 1\\
\underline1 & 1 & 1 & 1 & 1\\
\underline1 & 1 & 1 & 1 & 1\\
\underline1 & 1 & 1 & 1 & 1
\end{bmatrix}$
\caption{Some of the connectivities $C_0,C_1,\dots,C_N$ for $h=5$. Here, $U=10$ and $N=15$. In each~$C_t$, we have underlined every `fresh' $1$, i.e., every $1$ that is not also present in $C_{t-1}$ ($t\neq0$).}
\label{fig:Ci}
\end{figure}

\begin{figure}[b] \centering
$\;E_6{=}
\begin{bmatrix}
0 & 0 & 0 & 0 & 0\\
0 & 0 & 0 & 1 & 0\\
0 & 0 & 0 & 0 & 0\\
0 & 0 & 0 & 0 & 0\\
0 & 0 & 0 & 0 & 0\\ 
\end{bmatrix}$
$\;E_8{=}
\begin{bmatrix}
0 & 0 & 0 & 0 & 0\\
0 & 0 & 1 & 0 & 0\\
0 & 0 & 0 & 0 & 0\\
0 & 0 & 0 & 0 & 0\\
0 & 0 & 0 & 0 & 0\\ 
\end{bmatrix}$
$\;E_8'{=}
\begin{bmatrix}
0 & 0 & 0 & 0 & 0\\
0 & 0 & 1 & 1 & 1\\
0 & 0 & 0 & 0 & 0\\
0 & 0 & 0 & 0 & 0\\
0 & 0 & 0 & 0 & 0\\ 
\end{bmatrix}$
\\[2ex]
$D_{11}{=}
\begin{bmatrix}
0 & 0 & 0 & 0 & 0\\
0 & 0 & 0 & 0 & 0\\
0 & 0 & 0 & 0 & 0\\
0 & 0 & 0 & 0 & 0\\
0 & 0 & 0 & 1 & 0\\ 
\end{bmatrix}$
$D_{12}{=}
\begin{bmatrix}
0 & 0 & 0 & 0 & 0\\
0 & 0 & 0 & 0 & 0\\
0 & 0 & 0 & 0 & 0\\
0 & 0 & 1 & 0 & 0\\
0 & 0 & 1 & 0 & 0\\ 
\end{bmatrix}$
$D_{12}'{=}
\begin{bmatrix}
0 & 0 & 1 & 1 & 1\\
0 & 0 & 1 & 1 & 1\\
0 & 0 & 1 & 1 & 1\\
0 & 0 & 1 & 1 & 1\\
0 & 0 & 1 & 1 & 1\\ 
\end{bmatrix}$
\caption{Some of the auxiliary matrices for building $C_0,C_1,\dots,C_N$ for $h=5$.}
\label{fig:Ei-Di}
\end{figure}

\textit{In the first stage}, we flip the $U:=\tbinom h2$ bits of the strict upper triangle, \textit{one cell at a time}, from column~$h$ to column~$2$, and from bottom to top in each column: i.e., we start with cells $(h\mmm1,h),(h\mmm2,h),\dots,(1,h)$, continue with $(h\mmm2,h\mmm1),\dots,(1,h\mmm1)$, etc., all the way to $(2,3),(1,3)$, and~$(1,2)$.  (See upper row of Fig.\,\ref{fig:Ci}.) Specifically, if we let ``\textit{the $t$-th cell}'' refer to the $t$th cell $(i_t,j_t)$ in this sequence ($1\leq t\leq U$), then  
\begin{equation}\label{equ:Ct-stage1}
C_0:=\IDTYh 
\qquad\text{and}\qquad
C_t:=C_{t-1} + E_t  \qquad (1\leq t\leq U), 
\end{equation}
where $E_t\in\bbBh$ is the matrix with only the $t$-th cell~$(i_t,j_t)$ set to~$1$ (Fig.\,\ref{fig:Ei-Di}). An equivalent definition, which will be more helpful in our proofs later, is:
\begin{equation}\label{equ:Ct-stage1'}\tag{\ref{equ:Ct-stage1}$'$}
C_0:=\IDTYh 
\qquad\text{and}\qquad
C_t:=C_{t-1} + E_t' \qquad (1\leq t\leq U),
\end{equation}
where $E_t'\in\bbBh$ is the matrix obtained from $E_t$ by also placing~$1$s in the cells to the right of the $t$-th cell in the same row (so that the~$1$s in~$E_t'$ are exactly at $(i_t,j_t),(i_t,j_t+1),\dots,(i_t,h)$) (Fig.\,\ref{fig:Ei-Di}). Easily, these extra~$1$s are not changing~$C_t$, since $C_{t-1}$ already has~$1$s in those cells and the addition is Boolean. Either way, the end result~$C_U$ of this stage contains~$1$s  in exactly those cells that lie on and above the main diagonal (Fig.\,\ref{fig:Ei-Di}).

\textit{In the second stage}, we flip the bits of the strict lower triangle, \textit{one column at a time}, from column~$h\mmm1$ to column~$1$: that is, we start with the one cell $(h,h\mmm1)$, continue with the two cells $(h,h\mmm2),(h\mmm1,h\mmm2)$, etc., all the way to the $h\mmm2$ cells $(h,2),(h\mmm1,2)\dots,(3,2)$, and the $h\mmm1$ cells $(h,1),(h\mmm1,1),\dots,(2,1)$. (See lower row of Fig.\,\ref{fig:Ci}.) Specifically, if we let ``\textit{the $t$-th column}'' refer to the column $j_t:=h\mmm(t\mmm U)$ of the new $1$s of~$C_t$ ($U\ppp1\leq t\leq N\mmm1$), then  
\begin{equation}\label{equ:Ct-stage2}
C_t:=C_{t-1} + D_t \qquad (U\ppp1\leq t\leq N\mmm1),
\end{equation}
where $D_t\in\bbBh$ is the matrix that contains~$1$s exactly in those cells of the $t$th column that lie strictly below the diagonal (Fig.\,\ref{fig:Ei-Di}). An equivalent definition, which will be more helpful in our proofs later, is:
\begin{equation}\label{equ:Ct-stage2'}\tag{\ref{equ:Ct-stage2}$'$}
C_t:=C_{t-1} + D_t' \qquad (U\ppp1\leq t\leq N\mmm1),
\end{equation}
where $D_t'\in\bbBh$ is the matrix that contains~$1$s in all cells of the $t$th column and of the columns to its right (so that the~$1$s in~$D_t'$ are exactly all cells of columns $j_t,j_t\ppp1,\dots,h$) (Fig.\,\ref{fig:Ei-Di}). As before, these extra~$1$s do not change~$C_t$, as $C_{t-1}$ already has~$1$s in those cells and the addition is Boolean. Either way, the end result~$C_{N-1}$ of this stage contains~$1$s in all of the cells (Fig.\,\ref{fig:Ei-Di}).

\textit{In the third stage}, we flip the bits of the entire matrix, \textit{all at once}, from~$1$ back to~$0$, to get the zero matrix $C_N:=\ZEROh$. 

This concludes the definition of the sequence $C_0,C_1,\dots,C_U,C_{U\ppp1},\dots,C_{N\mmm1},C_N$. Note that the total number of steps is indeed $\binom h2+(h\mmm1)+1=\binom{h+1}2=N$, as promised.

\subsection{Some basic facts}
\label{sec:basic-facts}

Two important facts about the matrices $C_0,C_1,\dots,C_N$ are (i)~that they are idempotent; and (ii)~that the product of every two successive ones (in either order) produces the later one (Lemma~\ref{lem:properties-of-Ct}). To prove these facts, we must first prove some elementary facts about our auxiliary matrices, $E_t'$ and $D_t'$, and how these interact with the $C_t$ (Lemmas~\ref{lem:properties-of-Ct-with-Et} and~\ref{lem:properties-of-Ct-with-Dt}).

\begin{lemma}\label{lem:properties-of-Ct-with-Et}
For all $t=1,2,\dots,U$:\quad $(E_t')^2=\ZEROh$ \;and\; $C_{t-1}E_t' = E_t' = E_t'C_{t-1}$.
\end{lemma}

\begin{proof}
Fix~$t$ and let $(i,j):=(i_t,j_t)$ be ``the $t$-th cell''. Note that $i<j$, since cell~$(i,j)$ lies strictly above the main diagonal. Also note that $C_{t-1}$ has~$0$ in cell $(i,j)$; hence, the $1$s in its upper strict triangle are not covering all of the $j$th column; so, none of them is in the $i$-th column (since $i<j$); so, the $i$th column of $C_{t-1}$ has only a single~$1$, on the diagonal.  

Let $e_i\in\bbB^{h\times 1}$ be the column vector with a single~$1$ at cell~$i$; and $r_j\in\bbB^{1\times h}$ be the row vector with~$1$ in cells~$j,j\ppp1,\dots,h$ and $0$~everywhere else. We then easily verify the following:
\begin{itemize}

\item $E_t'=e_ir_j$. Indeed, the product is an $h\times h$ matrix where the only non-zero row is the $i$-th one (because of $e_i$'s single~$1$ in cell~$i$) and equals the row vector $r_j$, namely it has $1$s in columns~$j,j+1,\dots,h$ ---which is exactly the description of~$E_t'$. 

\item $r_j e_i = 0$. Indeed, the product is a $1\times 1$ matrix, namely a single Boolean value. And it equals~$0$, because $r_j$ has~$0$ in cell~$i$ (since $i<j$).

\item $C_{t-1}e_i=e_i$. Indeed, the product is a $h\times 1$ column vector, where the only non-zero cell is the $i$th one, because the $i$th row of $C_{t-1}$ is the only row with~$1$ in cell~$i$ (since, as observed above, the $i$-th column in $C_{t-1}$ has only a single~$1$, on the diagonal). 

\item $r_j C_{t-1}= r_j$. Indeed, the product is a $1\times h$ row vector, where the only non-zero cells are $j,j+1,\dots,h$, because columns $j,j+1,\dots,h$ are the only columns in $C_{t-1}$ that contain one or more~$1$s in cells $j,j\ppp1,\dots,h$ (since $C_{t-1}$ has $0$s in all cells strictly below the  diagonal). 
\end{itemize}
Using these observations, we easily compute:
\begin{gather*}    
(E_t')^2 
= (e_ir_j)^2 
= e_i r_j e_ir_j 
= e_i\cdot 0 \cdot r_j 
= \ZERO\cdot r_j 
= \ZEROh
\\[1ex]
C_{t-1}E_t' 
= C_{t-1} e_i r_j
= e_i r_j
= E_t' 
\qquad\qquad 
E_t'C_{t-1} 
= e_ir_j C_{t-1}
= e_ir_j
= E_t' \,,
\end{gather*}
where $\ZERO$ is the zero column vector and $\ZEROh$ is the zero $h\times h$ matrix.
\end{proof}

\begin{lemma}\label{lem:properties-of-Ct-with-Dt}
For all $t=U\ppp1,\dots,N\mmm1$:\quad $(D_t')^2=D_t'$ \;and\; $C_{t-1}D_t' = D_t' = D_t'C_{t-1}$.
\end{lemma}

\begin{proof}
Fix~$t$ and let~$j:=j_t=h\mmm(t\mmm U)$ be the index of the $t$-th column.

Let~$\ONE\in\bbB^{h\times 1}$ be the column vector containing only~$1$s; and (as in Lemma~\ref{lem:properties-of-Ct-with-Et}) $r_j\in\bbB^{1\times h}$ be the row vector with~$1$s exactly in cells~$j,j\ppp1,\dots,h$. We then easily verify the following:
\begin{itemize}

\item $D_t'=\ONE r_j$. Indeed, the product is an $h\times h$ matrix where every row is a copy of $r_j$, namely it has $1$s in columns~$j,j+1,\dots,h$ ---which is exactly the description of~$D_t'$. 

\item $r_j \ONE = 1$. Indeed, the product is a $1\times 1$ matrix, namely a single Boolean value. And it equals~$1$, because $r_j$ has at least one~$1$.

\item $C_{t-1}\ONE=\ONE$. Indeed, the product is a $h\times 1$ column vector, where a cell~$i$ contains~$1$ iff the $i$th row in $C_{t-1}$ has at least one~$1$; but every row of $C_{t-1}$ does (e.g., on the diagonal). 

\item $r_j C_{t-1}= r_j$. Indeed, the product is a $1\times h$ row vector, where the only non-zero cells are $j,j+1,\dots,h$, because columns $j,j+1,\dots,h$ are the only ones in $C_{t-1}$ with one or more~$1$s in the cells $j,j\ppp1,\dots,h$ (since $C_{t-1}$ has $0$s in all cells below the diagonal and left of column~$j$).%
\footnote{Note that, despite appearances, this fourth condition is not the same as the corresponding fourth condition in the proof of Lemma~\ref{lem:properties-of-Ct-with-Et}: $C_{t-1}$ there and here are different matrices, since $t\leq U$ there and $t\geq U\ppp1$ here.}
\end{itemize}
Using these observations, we easily compute:
\begin{gather*}    
(D_t')^2 
= (\ONE r_j)^2 
= \ONE r_j \ONE r_j 
= \ONE \cdot 1 \cdot r_j 
= \ONE r_j 
= D_t'
\\[1ex]
C_{t-1}D_t' 
= C_{t-1} \ONE r_j
= \ONE r_j
= D_t' 
\qquad\qquad 
D_t'C_{t-1} 
= \ONE r_j C_{t-1}
= \ONE r_j
= D_t' \,,
\end{gather*}
which covers all of the desired equalities.
\end{proof}

\begin{lemma}\label{lem:properties-of-Ct}
For all $t=0,1,\dots,N$:\quad $C_t^2 = C_t$ \;and\;\/ $C_{t-1}C_t = C_t = C_tC_{t-1}$ \;$($if~$t\neq0)$.
\end{lemma}

\begin{proof}
The idempotence of~$C_t$ is obvious for $t=N$: $C_N^2=\ZEROh^2=\ZEROh=C_N$. For~$t\leq N\mmm1$, we do induction on~$t$. In the base case ($t=0$), we have  $C_0^2=\IDTYh^2=\IDTYh=C_0$. In the inductive step ($1\leq t\leq N\mmm1$), we take cases according to whether $1\leq t\leq U$ (stage~1); or $U\ppp1\leq t\leq N\mmm1$ (stage~2). 

\textit{In stage~1} (cf.~\eqref{equ:Ct-stage1'}), we have $C_t^2=(C_{t-1}+E_t')^2= C_{t-1}^2 + C_{t-1}E_t' + E_t'C_{t-1} + (E_t')^2$. By the inductive hypothesis and Lemma~\ref{lem:properties-of-Ct-with-Et}, this becomes $C_{t-1} + E_t' + E'_t + \ZEROh$; which is just $C_{t-1}+E'_t$ (as the last three terms Boolean-add to~$E'_t$); which is just~$C_t$ (by~\eqref{equ:Ct-stage1'}).

\textit{In stage~2} (cf.~\eqref{equ:Ct-stage2'}), we have $C_t^2=(C_{t-1}+D_t')^2= C_{t-1}^2 + C_{t-1}D_t' + D_t'C_{t-1} + (D_t')^2$. By the inductive hypothesis and Lemma~\ref{lem:properties-of-Ct-with-Dt}, this becomes $C_{t-1} + D_t' + D'_t + D_t'$; which is just $C_{t-1}+D'_t$ (as the last three terms Boolean-add to~$D'_t$); which is just~$C_t$ (by~\eqref{equ:Ct-stage2'}).

This concludes the induction for the idempotence of $C_t$. 

For the equalities, we take cases according to whether $1\leq t\leq U$ (stage~1); or $U\ppp1\leq t\leq N\mmm1$ (stage~2); or $t=N$ (stage~3). 

\textit{In stage~1} (cf.~\eqref{equ:Ct-stage1'}), we have $C_{t-1}C_t=C_{t-1}(C_{t-1}+E_t')=C_{t-1}^2+C_{t-1}E_t'$. By idempotence and Lemma~\ref{lem:properties-of-Ct-with-Et}, this is $C_{t-1}+E_t'$, which is~$C_t$ (by~\eqref{equ:Ct-stage1'}). For the reverse product, we have 
$C_tC_{t-1}=(C_{t-1}+E_t')C_{t-1}=C_{t-1}^2+E_t'C_{t-1}$, which (by idempotence and Lemma~\ref{lem:properties-of-Ct-with-Et} again) is just $C_{t-1}+E'_t$, namely~$C_t$ again (by~\eqref{equ:Ct-stage1'}). Overall, $C_{t-1}C_t=C_t=C_tC_{t-1}$, as desired.

\textit{In stage~2} (cf.~\eqref{equ:Ct-stage2'}), we have $C_{t-1}C_t=C_{t-1}(C_{t-1}+D_t')=C_{t-1}^2+C_{t-1}D_t'$. By idempotence and Lemma~\ref{lem:properties-of-Ct-with-Dt}, this is $C_{t-1}+D_t'$, which is~$C_t$ (by~\eqref{equ:Ct-stage2'}). For the reverse product, we have 
$C_tC_{t-1}=(C_{t-1}+D_t')C_{t-1}=C_{t-1}^2+D_t'C_{t-1}$, which (by idempotence and Lemma~\ref{lem:properties-of-Ct-with-Dt} again) is just $C_{t-1}+D'_t$, namely~$C_t$ again (by~\eqref{equ:Ct-stage2'}). Overall, $C_{t-1}C_t=C_t=C_tC_{t-1}$, as desired.

\textit{In stage~3} ($t=N$), we have $C_{t-1}C_t=C_{t-1}\ZEROh=\ZEROh=C_t$ and $C_tC_{t-1}=\ZEROh C_{t-1}=\ZEROh=C_t$. Therefore, we again have $C_{t-1}C_t=C_t=C_tC_{t-1}$, as desired.
\end{proof}

\subsection{The desired facts}
\label{sec:desired-facts}

We are now ready to prove the three facts that we promised about the properties~$P_t$ of~\eqref{equ:P0,P1,...PN}.

\begin{lemma}\label{lem:every-Pt-is-smooth}
For all $t=0,\dots,N$, property~$P_t$ is smooth.
\end{lemma}

\begin{proof}
Pick any $x\kkk z\in P_t$. Let $y$ be any string with identity connectivity (e.g., the string consisting of the one symbol $\{(i,i)\mid i\in[h]\}$). Then the  connectivity of $xyz$ is $C(xyz) = C(x)C(y)C(z) = C_t \IDTYh C_t = C_t^2 = C_t$, where we used idempotence (Lemma~\ref{lem:properties-of-Ct}). So, $xyz\in P_t$.
\end{proof}

\begin{lemma}\label{lem:successive-Pi-separated-by-OWL}
For all $t=1,\dots,N$, properties~$P_{t-1},P_t$ are separated by~$\owl_h$.
\end{lemma}

\begin{proof}
Follows directly from Lemma~\ref{lem:distinct-conns-separated}, since $C_{t-1}\neq C_t$ for all~$t=1,\dots,N$. Indeed, the two connectivities $C_{t-1}$ and $C_t$ differ: either on the $t$th cell, if $1\leq t\leq U$ (stage~1); or on the $t$th column, if $U\ppp1\leq t\leq N\mmm1$ (stage~2); or on all cells, if $t=N$ (stage~3).
\end{proof}

\begin{lemma}\label{lem:Pt-1-has-suffix-of-choice-into-Pt}
For all $t=1,\dots,N$, property $P_{t-1}$~has suffix of choice into~$P_t$.    
\end{lemma}

\begin{proof}
Pick any $v\in P_{t-1}$. Let $u$~be arbitrary in~$P_t$. Pick any~$x\in P_{t-1}\cup P_t$. We compute:
\[
C(xuv)=C(x)C(u)C(v)=C(x)C_tC_{t-1}=C(x)C_t \,,
\] 
where the last equality uses Lemma~\ref{lem:properties-of-Ct}. Now, since $x$ is in~$P_{t-1}$ or $P_t$, the connectivity $C(x)C_t$ is either $C_{t-1}C_t$ or $C_t^2$; both of which equal~$C_t$, again by Lemma~\ref{lem:properties-of-Ct}. Therefore, $xuv\in P_t$.
\end{proof}

\subsection{The theorem}
\label{sec:the-theorem}

We are now ready to carefully state our main result and summarize its proof.

\begin{theorem}\label{thm:main}
Let~$h\geq1$. Let~$M$ be any \tdfa\ that solves $\owl_h$. Then $M$ has at least $\frac14h(h+1)$~states. 
\end{theorem}

\begin{proof}
For the given~$h$ and for $N=\binom{h+1}{2}$, let $P_0,P_1,\dots,P_N \subs\ggS_h\s$ be the $1\ppp N$ properties of~\eqref{equ:P0,P1,...PN}, which are all non-empty (clearly) and smooth (Lemma~\ref{lem:every-Pt-is-smooth}). For every $t=1,2,\dots,N$, we know that $P_{t-1}$ and~$P_t$ are separated by $\owl_h$ (Lemma~\ref{lem:successive-Pi-separated-by-OWL}); and that $P_{t-1}$ has suffix of choice into~$P_t$ (Lemma~\ref{lem:Pt-1-has-suffix-of-choice-into-Pt}).

For the given $M$ and for $t=0,1,\dots,N$, let $a_t:=a(P_t)$ and $b_t:=b(P_t)$ be the (well-defined, since $P_t$ is smooth) \tagLR-size and \tagRL-size of~$P_t$ with respect to~$M$ (Def.~\ref{def:LR-size} and discussion preceding Lemma~\ref{lem:suffix-of-choice-affects-property-exit-sizes}), which clearly satisfy $0\leq a_t,b_t\leq|Q|$, where $Q$~is the set of states of~$M$. 

Now, fix any $t=1,2,\dots,N$. We easily see that $P_{t-1}$ and $P_t$ satisfy the conditions of our Main Lemma (Lemma~\ref{lem:main}). Hence, $a_{t-1}\geq a_t$; and $b_{t-1}\geq b_t$; and at least one of these two inequalities is strict.

Overall, $M$~implies two weakly decreasing sequences $a_0\geq a_1\geq \cdots \geq a_N$ and
$b_0\geq b_1\geq \cdots \geq b_N$, where at least~$N$ of all inequalities are strict. Clearly then, one of these sequences hosts at least~$\frac12N$ strict inequalities, and thus contains at least~$1+\frac12N$ distinct integers. Since these all lie in $\{0,1,\dots,|Q|\}$, it follows that $1+\frac12N\leq1+|Q|$, which implies $|Q|\geq\frac14h(h+1)$.
\end{proof}

\section{Conclusion}
\label{sec:conclusion}

We proved that every \tdfa\ solving $\owl_h$ needs $\ggV(h^2)$ states and thus, in general, the conversion of \onfas\ to \tdfas\ incurs a quadratic blow-up in the number of states. Although this simply matches the already known lower bound by Chrobak~\cite{ch86,to09} and only slightly exceeds that of~\cite{ka18}, it results from an analysis of \tdfa\ computations which are much richer than the ones realized on unary or on $3$-long inputs.  

We conjecture that the sequence of smooth properties constructed in Sect.~\ref{sec:the-lower-bound} is of maximum possible length. Namely that, in every sequence $\emptyset\neq P_0,P_1,\dots,P_m\subs\ggS_h\s$, if every two successive properties $P_{t-1},P_t$ satisfy the conditions of our Main Lemma (Lemma~\ref{lem:main}), then $m\leq\binom{h+1}{2}$. 

If this is indeed the case, then our lemma in its present form seems not to be helpful in discovering super-quadratic lower bounds for \tdfas\ against~\owl.   
Still, we are hopeful that our Main Lemma might help bound the size of \tdfas\ in other settings. The lack of general lower-bounding tools is a well-known major deficiency in the study of \tdfa~computations. In this work, we deliberately tried to extract from our syllogisms a lemma that is general enough to potentially serve as such a tool.

This work used no artificial intelligence.

\nocite{bokoun18}
\newif\ifshortbibvenues

\newcommand{\newconfsl}[3]{\newcommand{#1}{Proceedings of \ifshortbibvenues #2\else the #3\fi}}
\newcommand{\newconfl }[2]{\newcommand{#1}{Proceedings of the #2}}

\newconfsl{\confAFL   }{AFL}   {{I}nternational {C}onference on {A}utomata and {F}ormal {L}anguages}
\newconfsl{\confCIAA  }{CIAA}  {{C}onference on {I}mplementation and {A}pplication of {A}utomata}
\newconfsl{\confCiE   }{CiE}   {{C}onference of {C}omputability in {E}urope}
\newconfsl{\confCCC   }{CCC}   {{C}onference on {C}omputational {C}omplexity}
\newconfsl{\confCSR   }{CSR}   {{I}nternational {C}omputer {S}cience {S}ymposium in {R}ussia}
\newconfsl{\confDCFS  }{DCFS}  {{W}orkshop on {D}escriptional {C}omplexity of {F}ormal {S}ystems}
\newconfsl{\confDLT   }{DLT}   {{I}nternational {C}onference on {D}evelopments in {L}anguage {T}heory}
\newconfsl{\confFCT   }{FCT}   {{I}nternational {S}ymposium on {F}undamentals of {C}omputation {T}heory}
\newconfsl{\confFOCS  }{FOCS}  {{S}ymposium on the {F}oundations of {C}omputer {S}cience}
\newconfsl{\confICALP }{ICALP} {{I}nternational {C}olloquium on {A}utomata, {L}anguages, and {P}rogramming}
\newconfsl{\confICIP  }{ICIP}  {{I}nternational {C}onference on {I}mage {P}rocessing}
\newconfsl{\confIFIP  }{IFIP}  {{I}nternational {F}ederation for {I}nformation {P}rocessing}
\newconfsl{\confISTCS }{ISTCS} {{I}srael {S}ymposium on the {T}heory of {C}omputing {S}ystems}
\newconfsl{\confLATA  }{LATA}  {{L}anguage and {A}utomata {T}heory and {A}pplications}
\newconfsl{\confMBB   }{MBB}   {{G}raduate {S}tudent {C}onference on {T}he {N}ature of {T}hought}
\newconfsl{\confMFCS  }{MFCS}  {{I}nternational {S}ymposium on {M}athematical {F}oundations of {C}omputer {S}cience}
\newconfsl{\confNCMA  }{NCMA}  {{I}nternational {W}orkshop on {N}on-{C}lassical {M}odels for {A}utomata and {A}pplications}
\newconfsl{\confSAGA  }{SAGA}  {{I}nternational {S}ymposium on {S}tochastic {A}lgorithms: {F}oundations and {A}pplications}
\newconfsl{\confSOAM  }{SOAM}  {{S}ymposium on {A}pplied {M}athematics}
\newconfsl{\confSODA  }{SODA}  {{S}ymposium on {D}iscrete {A}lgorithms}
\newconfsl{\confSOFSEM}{SOFSEM}{{I}nternational {C}onference on {C}urrent {T}rends in {T}heory and {P}ractice of {I}nformatics}
\newconfsl{\confSTACS }{STACS} {{S}ymposium on {T}heoretical {A}spects of {C}omputer {S}cience}
\newconfsl{\confSTOC  }{STOC}  {{S}ymposium on the {T}heory of {C}omputing}
\newconfl {\confSWAT  }{{S}ymposium on {S}witching and {A}utomata {T}heory}
\newconfl {\confSWCT  }{{S}ymposium on {S}witching {C}ircuit {T}heory and {L}ogical {D}esign}
\newconfl {\confIMYCS }{{I}nternational {M}eeting of {Y}oung {C}omputer {S}cientists}

\newcommand{\bibvenueskey}{\ifbibvenuesshort{\footnotesize 
AFL    is the \textit{International Conference on Automata and Formal Languages}
CIAA   is the \textit{Conference on Implementation and Application of Automata}.
CiE    is the \textit{Conference of Computability in Europe}
CCC    is the \textit{Conference on Computational Complexity}.
CSR    is the \textit{International Computer Science Symposium in Russia}.
DCFS   is the \textit{Workshop on Descriptional Complexity of Formal Systems}.
DLT    is the \textit{International Conference on Developments in Language Theory}.
FCT    is the \textit{International Symposium on Fundamentals of Computation Theory}.
FOCS   is the \textit{Symposium on the Foundations of Computer Science}.
ICALP  is the \textit{International Colloquium on Automata, Languages, and Programming}.
ICIP   is the \textit{International Conference on Image Processing}.
IFIP   is the \textit{International Federation for Information Processing}.
ISTCS  is the \textit{Israel Symposium on the Theory of Computing Systems}.
LATA   is the \textit{Conference on Language and Automata Theory and Applications}
MBB    is the \textit{Graduate Student Conference on The Nature of Thought}.
MFCS   is the \textit{International Symposium on Mathematical Foundations of Computer Science}.
SAGA   is the \textit{International Symposium on Stochastic Algorithms: Foundations and Applications}.
SOAM   is the \textit{Symposium on Applied Mathematics}.
SODA   is the \textit{Symposium on Discrete Algorithms}.
SOFSEM is the \textit{International Conference on Current Trends in Theory and Practice of Informatics}
STACS  is the \textit{Symposium on Theoretical Aspects of Computer Science}.
STOC   is the \textit{Symposium on the Theory of Computing}.
}\else\fi}
\shortbibvenuestrue 
\bibliographystyle{plain}
\bibliography{biblio/string,biblio/collections,biblio/personal,biblio/computation}

\end{document}